\documentclass[11pt]{article}
\usepackage{hyperref}
\usepackage{times}  
\usepackage{mathpazo}
\usepackage{amssymb,amsmath,amsthm}
\usepackage{epsfig}

\newtheorem{theorem}{Theorem}[section]
\newtheorem{proposition}[theorem]{Proposition}
\newtheorem{definition}[theorem]{Definition}

\newtheorem{lemma}[theorem]{Lemma}
\newtheorem{conjecture}[theorem]{Conjecture}

\newtheorem{corollary}[theorem]{Corollary}
\newtheorem{fact}[theorem]{Fact}

\newtheorem{remark}[theorem]{Remark}

\newcommand{\qedsymb}{\hfill{\rule{2mm}{2mm}}}
\renewenvironment{proof}[1][]{\begin{trivlist}
\item[\hspace{\labelsep}{\bf\noindent Proof#1:\/}] }{\qedsymb\end{trivlist}}

\def\calG{{\cal G}}

\def\calW{{\cal W}}

\def\Z{{\mathbb{Z}}}
\def\R{\mathbb{R}}

\def\mod{\mbox{mod}}

\newcommand\Kneser[3]{K^<(#1,#2,#3)}

\newcommand{\NP}{\mathsf{NP}}

\newcommand{\od}{\overline{\xi}}

\newcommand{\eps}{\epsilon}
\renewcommand{\epsilon}{\varepsilon}

\newcommand{\rank}{\mathop{\mathrm{rank}}}
\newcommand{\minrank}{\mathop{\mathrm{minrk}}}

\newcommand{\linspan}{\mathop{\mathrm{span}}}
\newcommand{\Fset}{\mathbb{F}}         

\newcommand{\vchrom}{{\chi_\mathrm{v}}}

\begin{document}

\title{{\bf The (Generalized) Orthogonality Dimension of (Generalized) Kneser Graphs: Bounds and Applications}}

\author{
Alexander Golovnev\thanks{Georgetown University, Washington, DC, USA. Partially supported by a Rabin Postdoctoral Fellowship.}
\and
Ishay Haviv\thanks{School of Computer Science, The Academic College of Tel Aviv-Yaffo, Tel Aviv 61083, Israel.
Research supported in part by the Israel Science Foundation (grant No. 1218/20).
}
}

\date{}

\maketitle

\begin{abstract}
The {\em orthogonality dimension} of a graph $G=(V,E)$ over a field $\Fset$ is the smallest integer $t$ for which there exists an assignment of a vector $u_v \in \Fset^t$ with $\langle u_v,u_v \rangle \neq 0$ to every vertex $v \in V$, such that $\langle u_v, u_{v'} \rangle = 0$ whenever $v$ and $v'$ are adjacent vertices in $G$.
The study of the orthogonality dimension of graphs is motivated by various applications in information theory and in theoretical computer science.
The contribution of the present work is two-fold.

First, we prove that there exists a constant $c$ such that for every sufficiently large integer $t$, it is $\NP$-hard to decide whether the orthogonality dimension of an input graph over $\R$ is at most $t$ or at least $3t/2-c$.
At the heart of the proof lies a geometric result, which might be of independent interest, on a {\em generalization} of the orthogonality dimension parameter for the family of {\em Kneser graphs}, analogously to a long-standing conjecture of Stahl (J.~Comb.~Theo.~Ser.~B,~1976).

Second, we study the smallest possible orthogonality dimension over finite fields of the complement of graphs that do not contain certain fixed subgraphs.
In particular, we provide an explicit construction of triangle-free $n$-vertex graphs whose complement has orthogonality dimension over the binary field at most $n^{1-\delta}$ for some constant $\delta >0$.
Our results involve constructions from the family of {\em generalized Kneser graphs} and they are motivated by the rigidity approach to circuit lower bounds.
We use them to answer a couple of questions raised by Codenotti, Pudl{\'{a}}k, and Resta (Theor.~Comput.~Sci.,~2000), and in particular, to disprove their Odd Alternating Cycle Conjecture over every finite field.
\end{abstract}

\section{Introduction}

A $t$-dimensional {\em orthogonal representation} of a graph $G=(V,E)$ over a field $\Fset$ is an assignment of a vector $u_v \in \Fset^t$ with $\langle u_v,u_v \rangle \neq 0$ to every vertex $v \in V$, such that $\langle u_v, u_{v'} \rangle = 0$ whenever $v$ and $v'$ are adjacent vertices in $G$. The {\em orthogonality dimension} of a graph $G$ over $\Fset$, denoted by $\od(G, \Fset)$, is the smallest integer $t$ for which there exists a $t$-dimensional orthogonal representation of $G$ over~$\Fset$.\footnote{Orthogonal representations of graphs are sometimes defined in the literature as orthogonal representations of the complement, namely, the definition requires vectors associated with {\em non-adjacent} vertices to be orthogonal. We have decided to use here the other definition, but one may view the notation $\od(G,\Fset)$ as standing for $\xi(\overline{G},\Fset)$.}
The orthogonality dimension parameter is closely related to several other well-studied graph parameters. In particular, for every graph $G$ and every field $\Fset$, $\od(G, \Fset)$ is sandwiched between the clique number and the chromatic number of $G$, that is, $\omega(G) \leq \od(G, \Fset) \leq \chi(G)$.

Orthogonal representations of graphs have been found useful over the years for various applications in information theory and in theoretical computer science.
They were originally introduced over the real field in a seminal work of Lov\'asz~\cite{Lovasz79}, where they were used to define the influential Lov\'asz $\vartheta$-function. The latter was used in~\cite{Lovasz79} to determine the Shannon capacity, a notoriously difficult information-theoretic graph parameter, of the cycle on five vertices, and in the last decades it was successfully applied in algorithmic and combinatorial results (see, e.g.,~\cite{knuth1994sandwich,Feige97,AlonK98}).
The orthogonality dimension of graphs plays an important role in several areas of computational complexity.
Over finite fields, the orthogonality dimension and its extension due to Haemers~\cite{Haemers79} to a graph parameter called {\em minrank} have attracted a significant attention in circuit complexity, and more specifically, in the study of Valiant's rigidity approach to circuit lower bounds~\cite{Valiant77} (see, e.g.,~\cite{CodenottiPR00,Riis07,GolovnevRW17}).
Over the complex field, the orthogonality dimension was used in a characterization of the quantum communication complexity of promise equality problems~\cite{deWolfThesis,BrietBLPS15,BrietZ17} and in the study of the quantum chromatic number~\cite{CameronMNSW07,ScarpaS12}.
The orthogonality dimension parameter was also investigated in the contexts of hardness of approximation~\cite{Peeters96,LangbergS08}, integrality gaps for linear programming~\cite{HuTS17,Haviv18topo}, and algorithms based on semi-definite programming~\cite{ChlamtacH14,HavivMFCS19}.

The present work studies two aspects of the orthogonality dimension of graphs.
First, we prove an $\NP$-hardness result for approximating the orthogonality dimension of graphs over the real field $\R$.
At the heart of the proof lies a geometric result, which might be of independent interest, on a {\em generalization} of the orthogonality dimension parameter for the family of {\em Kneser graphs}, analogously to a long-standing graph-theoretic conjecture due to Stahl~\cite{Stahl76}.
The second aspect of the orthogonality dimension parameter considered in this work, motivated by the area of circuit complexity, is that of determining the smallest possible orthogonality dimension  over finite fields of the complement of graphs that do not contain certain fixed subgraphs. In this context, we prove a new bound on the minrank parameter over finite fields for the family of {\em generalized Kneser graphs}.
The bound is used to settle a couple of questions raised by Codenotti, Pudl{\'{a}}k, and Resta in~\cite{CodenottiPR00} and to disprove their Odd Alternating Cycle Conjecture over every finite field.

\subsection{Our Contribution}

\subsubsection{The Generalized Orthogonality Dimension of Kneser Graphs}

We start by considering the computational hardness of determining the orthogonality dimension of graphs over the real field $\R$.
The challenge of understanding the hardness of this parameter was posed already in the late eighties by Lov\'asz, Saks, and Schrijver~\cite{LovaszSS89} (see also~\cite[Chapter~10]{LovaszBook}), and yet, the problem is far from being well-understood.
It is easy to see that deciding whether an input graph $G$ satisfies $\od(G,\R) \leq t$ can be solved in polynomial running-time for $t \in \{1,2\}$, and Peeters~\cite{Peeters96} has shown that it is $\NP$-hard for $t \geq 3$.
His result is known to imply that for every $t \geq 6$ it is $\NP$-hard to decide whether an input graph $G$ satisfies $\od(G,\R) \leq t$ or $\od(G,\R) \geq \lceil 4t/3 \rceil$ (see~\cite{HavivMFCS19}).
In the current work, we improve on the $4/3$ multiplicative gap and prove the following.

\begin{theorem}\label{thm:IntroHardness}
There exists a constant $c$ such that for every sufficiently large integer $t$, it is $\NP$-hard to decide whether an input graph $G$ satisfies $\od(G,\R) \leq t$ or $\od(G,\R) \geq 3t/2-c$.
\end{theorem}

It is worth noting that in order to obtain hardness results for the orthogonality dimension parameter, it is natural to employ known hardness results regarding the closely related chromatic number of graphs. Indeed, it is easy to verify (see, e.g.,~\cite{HavivMFCS19}) that every graph $G$ satisfies
\[\log_3 \chi(G) \leq \od(G,\R) \leq \chi(G),\]
hence hardness of deciding whether an input graph $G$ satisfies $\chi(G) \leq t_1$ or $\chi(G) \geq t_2$ immediately implies the hardness of deciding whether it satisfies $\od(G,\R) \leq t_1$ or $\od(G,\R) \geq \log_3 t_2$. In particular, a result of Dinur, Mossel, and Regev~\cite{DinurMR06} on the hardness of the chromatic number implies that assuming a certain variant of the unique games conjecture, deciding whether a given graph $G$ satisfies $\od(G,\R) \leq 3$ or $\od(G,\R) \geq t$ is $\NP$-hard for every $t \geq 4$. However, if one is interested in standard $\NP$-hardness for the orthogonality dimension, the state of the art for the hardness of the chromatic number does not seem to imply any hardness results, despite some remarkable recent progress~\cite{BKO19,WZ20}.
Moreover, most hardness proofs for the chromatic number crucially use the fact that an upper bound on the independence number of a graph implies a strong lower bound on its chromatic number (namely, $\chi(G) \geq \frac{|V(G)|}{\alpha(G)}$), whereas an analogue of such a statement for the orthogonality dimension does not hold in general (see, e.g.,~\cite[Proposition~2.2]{HavivMFCS19}).

One technique for proving hardness results for the chromatic number that can be applied for the orthogonality dimension is that of Garey and Johnson~\cite{GareyJ76a}, who have related hardness of graph coloring to the {\em multichromatic numbers of Kneser graphs}. The $k$th multichromatic number of a graph~$G$, denoted by $\chi_k(G)$, is the smallest number of colors needed in order to assign to every vertex of $G$ a set of $k$ colors so that adjacent vertices are assigned to disjoint sets. Notice that $\chi_1(G)$ is simply the standard chromatic number $\chi(G)$. The family of Kneser graphs is defined as follows.
\begin{definition}[Kneser Graphs]\label{def:Kneser}
For integers $d \geq 2s$, the {\em Kneser graph} $K(d,s)$ is the graph whose vertices are all the $s$-subsets of $[d]=\{1,\ldots,d\}$, where two sets are adjacent if they are disjoint.
\end{definition}
\noindent
Note that the multichromatic numbers can be defined in terms of Kneser graphs, namely, $\chi_k(G)$ is the smallest integer $d$ for which there exists a homomorphism from $G$ to $K(d,k)$.

In the seventies, Stahl~\cite{Stahl76} has made the following conjecture.
\begin{conjecture}[Stahl's Conjecture~\cite{Stahl76}]\label{conj:Stahl}
For all integers $k$ and $d \geq 2s$,
\[\chi_k(K(d,s)) = \Big \lceil \frac{k}{s} \Big \rceil \cdot (d-2s)+2k.\]
\end{conjecture}
\noindent
Stahl's conjecture has received a significant attention in the literature over the years. Even very recently, it was related to the well-known recently disproved Hedetniemi's conjecture~\cite{TradifZ19}.
Nevertheless, more than forty years since it was proposed, Stahl's conjecture is still open.
It is known that the right-hand side in Conjecture~\ref{conj:Stahl} forms an upper bound on $\chi_k(K(d,s))$, and that this bound is tight up to an additive constant that depends solely on $s$~\cite{CGJ78,Stahl98}.
The precise statement of the conjecture was confirmed only for a few special cases. This includes the case of $k=1$ proved by Lov\'asz~\cite{LovaszKneser}, the cases of $s \leq 2$, $k \leq s$, $d=2s+1$, and $k$ divisible by $s$ proved by Stahl~\cite{Stahl76,Stahl98}, and the case of $s=3$ and $k=4$ proved by Garey and Johnson~\cite{GareyJ76a} (extended to $s=3$ with any $k$ in~\cite{Stahl98}). The result of~\cite{GareyJ76a} was combined there with a simple reduction to show that for every $t \geq 6$, it is $\NP$-hard to decide whether a given graph $G$ satisfies $\chi(G) \leq t$ or $\chi(G) \geq 2t-4$.

The recent work~\cite{HavivMFCS19} has suggested to borrow the reduction of~\cite{GareyJ76a} to prove hardness results for the orthogonality dimension parameter. This approach requires the following generalization of orthogonal representations of graphs over the reals.

\begin{definition}[Orthogonal Subspace Representation\footnote{Over the complex field, the definition is equivalent to the notion of a projective representation from~\cite[Definition~6.1]{manvcinska2016quantum}.}]\label{def:ortho_k-subspace}
A $t$-dimensional {\em orthogonal $k$-subspace representation} of a graph $G=(V,E)$ is an assignment of a subspace $U_v \subseteq \R^t$ with $\dim (U_v)=k$ to every vertex $v \in V$, such that the subspaces $U_v$ and $U_{v'}$ are orthogonal whenever $v$ and $v'$ are adjacent in $G$. For a graph $G$, let $\od_k(G,\R)$ denote the smallest integer $t$ for which there exists a $t$-dimensional orthogonal $k$-subspace representation of $G$.
\end{definition}
\noindent
Note that for $k=1$, Definition~\ref{def:ortho_k-subspace} coincides with the orthogonality dimension over the reals, and that for every graph $G$ and every $k$ it holds that $\od_k(G,\R) \leq \chi_k(G)$.

A combination of the hardness result of Peeters~\cite{Peeters96} and the reduction of~\cite{GareyJ76a} implies the following.

\begin{proposition}[{\cite[Theorem~1.3]{HavivMFCS19}}]\label{prop:xi3vs4}
For every graph $F$, it is $\NP$-hard to decide whether an input graph $G$ satisfies $\od(G,\R) \leq \od_3(F,\R)$ or $\od(G,\R) \geq \od_4(F,\R)$.
\end{proposition}
\noindent
With Proposition~\ref{prop:xi3vs4} in hand, it is of interest to find graphs $F$ with a large gap between $\od_3(F,\R)$ and $\od_4(F,\R)$.
In light of Conjecture~\ref{conj:Stahl}, it is natural to consider the generalized orthogonality dimension parameters for the family of Kneser graphs.
For $k=1$, it was shown in~\cite{Haviv18topo} that the standard chromatic number and the standard orthogonality dimension over $\R$ coincide on all Kneser graphs. In addition, a result of Bukh and Cox~\cite[Proposition~23]{BukhC18} implies that for every $d \geq 2s$ and every~$k$, $\od_k(K(d,s),\R) \geq kd/s$. This implies that the $k$th chromatic number and the $k$th orthogonality dimension over $\R$ coincide on $K(d,s)$ whenever $k$ is divisible by $s$.

In this work we initiate a systematic study of the generalized orthogonality dimension parameters of Kneser graphs, analogously to Conjecture~\ref{conj:Stahl}.
Let us already mention that the arguments applied in the study of Stahl's conjecture do not seem to extend to our question.
The main reason is that the proofs in~\cite{Stahl76,GareyJ76a,CGJ78,Stahl98} use Hilton-Milner-type theorems to characterize the possible structures of the independent sets induced by generalized colorings of Kneser graphs, whereas in our setting, orthogonal subspace representations do not naturally induce large independent sets and the problem seems to require a more geometric approach.

The first non-trivial case is that of Kneser graphs $K(d,s)$ with $s=2$, for which we show that the generalized orthogonality dimension parameters are equal to the multichromatic numbers.

\begin{theorem}\label{thm:Intro_s=2}
For all integers $k \geq 1$ and $d \geq 4$,
\[\od_{k}(K(d,2),\R) = \Big \lceil \frac{k}{2} \Big \rceil \cdot (d-4)+2k.\]
\end{theorem}

We proceed by considering a general $s \geq 3$ and prove the following lower bound.

\begin{theorem}\label{thm:Intro_general_s}
For every integers $k \geq s \geq 3$ there exists $c=c(s,k)$ such that for all integers $d \geq 2s$,
\[\od_{k}(K(d,s),\R) \geq \frac{k- \lceil \frac{k+1}{s} \rceil +1}{s-1} \cdot d-c.\]
\end{theorem}
\noindent
Note that for $k = \ell \cdot s -1$ where $\ell$ is an integer, the bound provided by Theorem~\ref{thm:Intro_general_s} is tight up to the additive constant $c$.
Indeed, in this case we get that there exists a constant $c$ such that for all integers $d \geq 2s$ it holds that
\[\ell \cdot d -c \leq \od_{\ell \cdot s-1}(K(d,s),\R) \leq \chi_{\ell \cdot s-1}(K(d,s)) \leq \ell \cdot d -2.\]
Note further that for the special case of $k=4$ and $s=3$, Theorem~\ref{thm:Intro_general_s} implies that there exists a constant $c$ such that $\od_4(K(d,3),\R) \geq 3d/2-c$ for every sufficiently large integer $d$. This, combined with Proposition~\ref{prop:xi3vs4} and the fact that $\od_3(K(d,3),\R) = d$, yields our hardness result Theorem~\ref{thm:IntroHardness}.

It will be interesting to figure out if the bounds given in Theorem~\ref{thm:Intro_general_s} can be tightened to the quantity given in the right-hand side of Conjecture~\ref{conj:Stahl}, at least up to an additive term independent of $d$.
In particular, it will be nice to decide whether for all integers $d \geq 6$ it holds that $\od_4(K(d,3),\R) = 2d-4$.
A positive answer would imply that for every $t \geq 6$, it is $\NP$-hard to decide whether an input graph $G$ satisfies $\od(G) \leq t$ or $\od(G) \geq 2t-4$. We remark, however, that the approach suggested by Proposition~\ref{prop:xi3vs4} for the hardness of the orthogonality dimension cannot yield a multiplicative hardness gap larger than $2$, as it is easy to see that every graph $F$ satisfies $\od_4(F,\R) \leq \od(F,\R) + \od_3(F,\R) \leq 2 \cdot \od_3(F,\R)$.

\subsubsection{The Orthogonality Dimension of Generalized Kneser Graphs}

We next consider the orthogonality dimension over finite fields of the complement of graphs that do not contain some fixed subgraphs.
In fact, in this context we consider an extension of the orthogonality dimension parameter, called minrank, that was introduced by Haemers in~\cite{Haemers79} and is defined as follows.

\begin{definition}[Minrank]\label{def:minrank}
Let $G=(V,E)$ be a directed graph on the vertex set $V=[n]$ and let $\Fset$ be a field.
We say that an $n$ by $n$ matrix $M$ over $\Fset$ {\em represents} $G$ if $M_{i,i} \neq 0$ for every $i \in V$, and $M_{i,j}=0$ for every distinct $i,j \in V$ such that $(i,j) \notin E$.
The {\em minrank} of $G$ over $\Fset$ is defined as
\[{\minrank}_\Fset(G) =  \min\{{\rank}_{\Fset}(M)\mid M \mbox{ represents }G\mbox{ over }\Fset\}.\]
The definition is naturally extended to (undirected) graphs by replacing every undirected edge with two oppositely directed edges.
\end{definition}
\noindent
Note that for every graph $G$ and every field $\Fset$, ${\minrank}_\Fset(\overline{G}) \leq \od(G,\Fset)$.\footnote{Indeed, given a $t$-dimensional orthogonal representation of an $n$-vertex graph $G$ over a field $\Fset$, consider the matrix $B \in \Fset^{n \times t}$ whose rows contain the vectors associated with the vertices of $G$. Then, the $n$ by $n$ matrix $B \cdot B^T$ represents $\overline{G}$ and has rank at most $t$ over $\Fset$, hence ${\minrank}_\Fset(\overline{G}) \leq t$.}

We consider here the question of whether there are graphs with no short odd cycles and yet low minrank over finite fields.
This question is motivated by the area of circuit complexity, and more specifically, by Valiant's approach to circuit lower bounds~\cite{Valiant77}, as described next.
The {\em rigidity} of an $n$ by $n$ matrix $M$ over a field $\Fset$ with respect to a given parameter $r$ is the smallest number of entries that one has to change in $M$ in order to reduce its rank over $\Fset$ to below $r$.
Roughly speaking, it was shown in~\cite{Valiant77} that $n$ by $n$ matrices with large rigidity for $r = \eps \cdot n$ where $\eps>0$ is a constant can be used to obtain superlinear lower bounds on the size of logarithmic depth arithmetic circuits computing linear transformations.
In 2000, Codenotti, Pudl{\'{a}}k, and Resta~\cite{CodenottiPR00} have proposed the {\em Odd Alternating Cycle Conjecture}, stated below.
By an {\em alternating odd cycle} we refer to a directed graph which forms an odd cycle when the orientation of the edges is ignored, and such that the orientation of the edges alternates with one exception.
\begin{conjecture}[The Odd Alternating Cycle Conjecture~\cite{CodenottiPR00}]\label{conj:alternating}
For every field $\Fset$ there exist $\eps >0$ and an odd integer $\ell$ such that every $n$-vertex directed graph $G$ with ${\minrank}_\Fset (G)  \leq \eps \cdot n$ contains an alternating cycle of length $\ell$.
\end{conjecture}

It was proved in~\cite{CodenottiPR00} that Conjecture~\ref{conj:alternating} implies, if true, that certain explicit circulant matrices have superlinear rigidity.
In contrast, for $\ell =3$ it was shown in~\cite{CodenottiPR00} that there are $n$-vertex (undirected) triangle-free graphs $G$ satisfying ${\minrank}_\Fset(G) \leq O(n^{3/4})$ for every field $\Fset$, and it was left open whether the statement of Conjecture~\ref{conj:alternating} may hold for larger values of $\ell$.
In the recent work~\cite{Haviv18free} the conjecture was disproved over the real field, but remained open for finite fields which are of special interest in circuit complexity.
For the orthogonality dimension over the binary field $\Fset_2$, it was shown in~\cite{CodenottiPR00} that there exist triangle-free $n$-vertex graphs $G$ satisfying $\od(\overline{G},\Fset_2) = n/4+2$. It was asked there whether every $n$-vertex graph $G$ satisfying $\od(\overline{G},\Fset_2) \leq n/4+1$ must contain a triangle.


In the current work we prove a new upper bound on the minrank parameter over finite fields of {\em generalized Kneser graphs}. In these graphs the vertices are all the $s$-subsets of a universe $[d]$, where two sets are adjacent if their intersection size is smaller than some integer $m$. Note that for $m=1$ we get the standard family of Kneser graphs (see Definition~\ref{def:Kneser}).
In the proof we modify and extend an argument of~\cite{Haviv18free}, which is based on linear spaces of multivariate polynomials, building on a previous work of Alon~\cite{AlonUnion98}. For the precise statement, see Theorem~\ref{thm:minrk_Kneser}.
We turn to describe several applications of our bound.

As a first application, we establish an explicit construction of graphs that do not contain short odd cycles and yet have low minrank over every finite field.
\begin{theorem}\label{thm:Intro_Cycles}
For every odd integer $\ell \geq 3$ there exists $\delta = \delta(\ell) >0$ such that for every sufficiently large integer $n$, there exists an $n$-vertex graph $G$ with no odd cycle of length at most $\ell$ such that for every finite field $\Fset$,
\[{\minrank}_{\Fset}(G) \leq n^{1-\delta}.\]
\end{theorem}
\noindent
Theorem~\ref{thm:Intro_Cycles} immediately implies that the Odd Alternating Cycle Conjecture is false over every finite field, even for undirected graphs.
This rules out the approach suggested in~\cite{CodenottiPR00} for lower bounds on the rigidity of certain circulant matrices and thus falls into the recent line of non-rigidity results based on the polynomial method (see, e.g.,~\cite{AlmanW17,dvir2019matrix,DvirL19}).
We note, however, that the general upper bound of~\cite{DvirL19} on the rigidity of $n \times n$ circulant matrices does not apply to the setting of parameters considered in~\cite{CodenottiPR00} (because in~\cite{DvirL19} the upper bound is $n^{1+\eps}$ for a constant $\eps>0$, whereas in~\cite{CodenottiPR00} the rigidity is only claimed to be $\Omega(n \cdot \log^\eps n)$ for a constant $\eps>0$).

We next consider the behavior of the orthogonality dimension over the binary field of the complement of triangle-free graphs.
It is relevant to mention here that in the proof of Theorem~\ref{thm:Intro_Cycles}, the matrices that imply the stated bound on the minrank are symmetric (see Remark~\ref{remark:symmetric}). For the binary field, this can be combined with a matrix decomposition result due to Lempel~\cite{Lempel75} to obtain the following theorem, which answers a question of~\cite{CodenottiPR00} negatively.
\begin{theorem}\label{thm:IntroTriangleFree}
There exists a constant $\delta >0$ such that for every sufficiently large integer $n$ there exists a triangle-free $n$-vertex graph $G$ such that $\od(\overline{G},\Fset_2) \leq n^{1-\delta}$.
\end{theorem}

The above result can also be stated in terms of nearly orthogonal systems.
For a field $\Fset$, a system of vectors in $\Fset^m$ is said to be {\em nearly orthogonal} if every vector of the system is not self-orthogonal and any set of three of them contains an orthogonal pair.
For the real field, it was proved by Rosenfeld~\cite{Rosenfeld91} that every nearly orthogonal system in $\R^m$ has size at most $2m$.
Theorem~\ref{thm:IntroTriangleFree} shows that the situation is quite different over the binary field. Namely, it implies that there exists a constant $\delta>0$ such that for infinitely many integers $m$ there exists a nearly orthogonal system in $\Fset_2^m$ of size at least $m^{1+\delta}$.

We finally mention that our bound on the minrank parameter of generalized Kneser graphs can be used to obtain graphs with a constant {\em vector chromatic number} $\vchrom$ (see Definition~\ref{def:chi_v}) whose complement has a polynomially large minrank over every finite field.
\begin{theorem}\label{thm:Intro_chi_v}
There exists a constant $\delta >0$ such that for infinitely many integers $n$ there exists an $n$-vertex graph $G$ such that $\vchrom(G) \leq 3$ and yet $\minrank_{\Fset}(\overline{G}) \geq n^\delta$ for every finite field $\Fset$.
\end{theorem}
\noindent
The interest in such graphs comes from the semidefinite programming algorithmic approach applied in~\cite{ChlamtacH14} for approximating the minrank parameter. As explained in~\cite{Haviv18}, such graphs imply a limitation on this approach, which is based on the constant vector chromatic number of the complement of the instances. Theorem~\ref{thm:Intro_chi_v} improves on~\cite[Theorem~1.3]{Haviv18} where the bound on the minrank is shown only for sufficiently large finite fields.

\subsection{Outline}
The rest of the paper is organized as follows.
In Section~\ref{sec:GenODofKneser}, we prove our bounds on the generalized orthogonality dimension parameters of Kneser graphs (Theorems~\ref{thm:Intro_s=2} and~\ref{thm:Intro_general_s}) and derive our hardness result (Theorem~\ref{thm:IntroHardness}).
In Section~\ref{sec:ODofGenKneser}, we prove our bound on the minrank parameter over finite fields of generalized Kneser graphs and deduce Theorems~\ref{thm:Intro_Cycles},~\ref{thm:IntroTriangleFree}, and~\ref{thm:Intro_chi_v}.

\section{The Generalized Orthogonality Dimension of Kneser Graphs}\label{sec:GenODofKneser}

In this section we study the generalized orthogonality dimension parameters of Kneser graphs, namely, the quantities $\od_k(K(d,s))$ (recall Definitions~\ref{def:Kneser} and~\ref{def:ortho_k-subspace}), and prove Theorems~\ref{thm:Intro_s=2} and~\ref{thm:Intro_general_s}.
We start with a linear algebra lemma that will be useful in our proofs.

\subsection{Linear Algebra Lemma}

\begin{lemma}\label{lem_cor:good_subspace}
Let $U$ be a subspace of $\R^t$ with $\dim(U) = \ell$, let $\calW$ be a finite collection of subspaces of $\R^t$, and let $\ell' \leq \ell$ be an integer satisfying $\dim(U \cap W) \leq \ell'$ for every $W \in \calW$. Then, there exists a subspace~$U'$ of $U$ with $\dim(U')=\ell-\ell'$ such that $\dim(U' \cap W) = 0$ for every $W \in \calW$.
\end{lemma}

Intuitively, given a subspace $U$ and a collection $\calW$ as in the lemma, a `random-like' subspace $U'$ of $U$ with dimension $\ell-\ell'$ is expected to have a trivial intersection with each of the subspaces of $\calW$, and thus to satisfy the assertion of the lemma.
To prove it formally, we use the following well-known fact.
\begin{fact}\label{fact:covering_by_subspaces}
Let $U$ be a subspace of $\R^t$, and let $\calW$ be a finite collection of proper subspaces of $U$.
Then, $\calW$ does not cover $U$, that is, there exists a vector $u \in U$ such that $u \notin W$ for every $W \in \calW$.
\end{fact}

\begin{proof}
We may assume without loss of generality that $U = \R^t$ (otherwise apply to $U$ and to the subspaces of $\calW$ an invertible linear transformation from $U$ to $\R^{\dim(U)}$) and that $\dim(W)=t-1$ for every $W \in \calW$ (otherwise replace $W$ with an arbitrary subspace of dimension $t-1$ that contains it).
Consider the set $S = \{ (1,\alpha,\alpha^2,\ldots,\alpha^{t-1}) \mid \alpha \in \R \} \subseteq \R^t$.
Every subspace $W \in \calW$ consists of all the points $x \in \R^t$ satisfying a certain linear equation $\sum_{i=1}^{t}{a_ix_i} = 0$, so the intersection $S \cap W$ is the set of all points $(1,\alpha,\alpha^2,\ldots,\alpha^{t-1})$ where $\alpha$ satisfies $\sum_{i=1}^{t}{a_i \alpha^{i-1}} = 0$. Since a polynomial of degree $t-1$ has at most $t-1$ zeros, it follows that $S \cap W$ is finite for every $W$, hence the union of the subspaces of the finite collection $\calW$ intersects $S$ at a finite number of points. This implies that $\calW$ does not cover $S$, and in particular, it does not cover the entire space $\R^t$.
\end{proof}

We use Fact~\ref{fact:covering_by_subspaces} to prove the following lemma.

\begin{lemma}\label{lemma:good_subspace}
Let $U$ be a subspace of $\R^t$ with $\dim(U) = \ell$, and let $\calW$ be a finite collection of subspaces of $U$.
Then for every integer $\ell' \leq \ell$ there exists a subspace $U'$ of $U$ with $\dim(U')=\ell'$ such that for every $W \in \calW$ it holds that \[\dim(U' \cap W) = \max \big ( 0,\dim(W)+\ell'-\ell \big ).\]
\end{lemma}

\begin{proof}
Let $U$ and $\calW$ be as in the statement of the lemma. It suffices to show that for every $\ell' \leq \ell$ there exists a subspace $U'$ of $U$ with $\dim(U')=\ell'$ such that for every $W \in \calW$ it holds that
\begin{equation}\label{eq:U'W_upper}
\dim(U' \cap W) \leq \max \big ( 0, \dim(W)+\ell'-\ell \big ).
\end{equation}
Indeed, such a subspace $U'$ also satisfies for every subspace $W \in \calW$ that $\dim(U' \cap W) \geq 0$ and that
\begin{eqnarray*}
\dim (U' \cap W) &=& \dim(U')+\dim(W)-\dim(U'+W) \geq \dim(U')+\dim(W)-\dim(U+W) \\
&=& \dim(U')+\dim(W)-\dim(U) = \dim(W) +\ell'-\ell,
\end{eqnarray*}
so the inequality in~\eqref{eq:U'W_upper} is in fact an equality.

We prove the existence of a subspace $U'$ satisfying~\eqref{eq:U'W_upper} by induction on $\ell'$.
For $\ell'=0$ the statement trivially holds for the choice $U' = \{0\}$.
Assume that for $1 \leq \ell' \leq \ell$ there exists a subspace $U''$ of $U$ with $\dim(U'')=\ell'-1$ such that for every $W \in \calW$ it holds that
\begin{eqnarray}\label{eq:induction_ell'}
\dim(U'' \cap W) \leq \max \big ( 0, \dim(W)+\ell'-\ell-1 \big ).
\end{eqnarray}
Consider the collection $\calW'$ of all the proper subspaces of $U$ in $\{U''+W \mid W \in \calW\}$.
By Fact~\ref{fact:covering_by_subspaces}, there exists a vector $u \in U$ that does not lie in any of the subspaces of $\calW'$. Define \[U' = U'' + \linspan(u).\]
Since $u \notin U''$ we have $\dim(U')=\dim(U'')+1=\ell'$.
Fix a subspace $W \in \calW$.
If $U''+W = U$ then
\[\dim(U'' \cap W) = \dim(U'')+\dim(W)-\dim(U''+W) = \dim(W) +\ell'-\ell-1,\]
and thus
\[\dim(U' \cap W) \leq \dim(U'' \cap W)+1 = \dim(W) +\ell'-\ell \leq \max \big ( 0, \dim(W) +\ell'-\ell \big ),\]
as required.
Otherwise, $U''+W$ is a proper subspace of $U$, hence $\dim(U''+W) \leq \ell-1$.
It follows that
\[\dim(U'' \cap W) = \dim(U'')+\dim(W)-\dim(U''+W) \geq \dim(W) +\ell'-\ell,\]
which using the induction hypothesis given in~\eqref{eq:induction_ell'} implies that $\dim(U'' \cap W)=0$.
Our choice of $u$ guarantees that $u \notin U''+W$, that is, $(u+U'') \cap W = \emptyset$, hence $\dim(U' \cap W)=\dim(U'' \cap W)=0$. In particular,
$\dim(U' \cap W) \leq \max \big (0, \dim(W) +\ell'-\ell \big )$, and we are done.
\end{proof}

Equipped with Lemma~\ref{lemma:good_subspace}, we are ready to derive Lemma~\ref{lem_cor:good_subspace}.

\begin{proof}[ of Lemma~\ref{lem_cor:good_subspace}]
Apply Lemma~\ref{lemma:good_subspace} to the subspace $U$ with the collection of its subspaces $\{ U \cap W \mid W \in \calW\}$. We obtain a subspace $U'$ of $U$ with $\dim(U')=\ell-\ell'$ such that for every $W \in \calW$ it holds that $\dim(U' \cap W) = \max \big ( 0,\dim(U \cap W)-\ell' \big ) = 0$, as desired.
\end{proof}

\subsection{The case $s=2$}

We turn to prove Theorem~\ref{thm:Intro_s=2} which determines the generalized orthogonality dimension parameters of Kneser graphs $K(d,s)$ for $s=2$.


\begin{proof}[ of Theorem~\ref{thm:Intro_s=2}]
Fix an integer $k \geq 1$.
For the upper bound, recall that for all integers $d \geq 4$ we have
\[\od_{k}(K(d,2),\R) \leq \chi_k(K(d,2)) = \Big \lceil \frac{k}{2} \Big \rceil \cdot (d-4)+2k.\]

For the lower bound, we consider the induced subgraph of $K(d,2)$, denoted by $K^{-}(d,2)$, obtained from $K(d,2)$ by removing one of its vertices, say, the vertex $\{1,2\}$. We turn to prove that for all integers $d \geq 4$ it holds that
\begin{eqnarray}\label{eq:K-,s=2}
\od_{k}(K^{-}(d,2),\R) \geq \Big \lceil \frac{k}{2} \Big \rceil \cdot (d-4)+2k,
\end{eqnarray}
which immediately implies the required lower bound on $\od_{k}(K(d,2),\R)$ as well.
To this end, we apply an induction on $d$.
For $d=4$, the graph $K(d,2)$ is a perfect matching on $6$ vertices, hence its subgraph $K^{-}(d,2)$ clearly contains an edge. Since every orthogonal $k$-subspace representation of this graph assigns to the vertices of this edge orthogonal $k$-subspaces it follows that \[\od_k(K^{-}(4,2),\R) \geq 2k,\]
as desired.
Now, fix some $d > 4$. Assuming that~\eqref{eq:K-,s=2} holds for $d-1$, we turn to prove it for $d$.

Recall that the vertex set $V$ of $K^{-}(d,2)$ consists of all the $2$-subsets of $[d]$ except $\{1,2\}$.
Let $(U_A)_{A \in V}$ be a $t$-dimensional orthogonal $k$-subspace representation of $K^{-}(d,2)$.
We proceed by considering the following two cases.

Assume first that there exists some $i \geq 4$ for which
\begin{eqnarray}\label{eq:U_{1,3}}
\dim(U_{\{1,3\}} \cap U_{\{1,i\}}) \geq \Big\lceil \frac{k}{2} \Big\rceil.
\end{eqnarray}
In this case, consider the induced subgraph of $K^{-}(d,2)$ on the vertex set $V'$ obtained from $V$ by removing the vertex $\{3,i\}$ and all the vertices that include the element $1$.
Notice that this subgraph is isomorphic to $K^{-}(d-1,2)$ and that every vertex of $V'$ is disjoint from either $\{1,3\}$ or from $\{1,i\}$ (or both).
This implies that the restriction $(U_A)_{A \in V'}$ of the given assignment to the vertices of $V'$ forms an orthogonal $k$-subspace representation of $K^{-}(d-1,2)$, all of whose subspaces lie in the subspace of $\R^t$ that is orthogonal to $U = U_{\{1,3\}} \cap U_{\{1,i\}}$.
By applying an orthogonal linear transformation from this subspace to $\R^{t-\dim(U)}$, we obtain that
\[ \od_k(K^{-}(d-1,2),\R) \leq t- \dim(U) \leq t- \Big\lceil \frac{k}{2} \Big\rceil,\]
where in the second inequality we have used~\eqref{eq:U_{1,3}}.
Using the induction hypothesis, this implies that
\[ t \geq \od_k(K^{-}(d-1,2),\R) + \Big\lceil \frac{k}{2} \Big\rceil \geq \Big \lceil \frac{k}{2} \Big \rceil \cdot (d-5)+2k + \Big\lceil \frac{k}{2} \Big\rceil = \Big \lceil \frac{k}{2} \Big \rceil \cdot (d-4)+2k,\]
and we are done.

We are left with the case where for every $i \geq 4$ it holds that
\[\dim(U_{\{1,3\}} \cap U_{\{1,i\}}) \leq \Big\lceil \frac{k}{2} \Big\rceil -1.\]
Apply Lemma~\ref{lem_cor:good_subspace} to the $k$-subspace $U_{\{1,3\}}$ and the collection $\{ U_{\{1,i\}} \mid 4 \leq i \leq d \}$.
It follows that there exists a subspace $U$ of $U_{\{1,3\}}$ with $\dim(U) = k- (\lceil \frac{k}{2} \rceil -1) \geq \lceil \frac{k}{2} \rceil$ such that for every $i \geq 4$ it holds that $\dim(U \cap U_{\{1,i\}}) = 0$.
Consider the induced subgraph of $K^{-}(d,2)$ on the vertex set $V'$ obtained from $V$ by removing the vertex $\{2,3\}$ and all the vertices that include the element $1$. As before, this subgraph is isomorphic to the graph $K^{-}(d-1,2)$.

We define an orthogonal $k$-subspace representation of this graph as follows.
Let $B$ be a set in $V'$.
If $3 \notin B$ we define $\widetilde{U}_B = U_B$.
Otherwise we have $B = \{3,i\}$ for some $i \geq 4$, and we let $\widetilde{U}_{\{3,i\}}$ be the projection of $U_{\{1,i\}}$ to the subspace of $\R^t$ that is orthogonal to $U$.
Note that the fact that $\dim(U \cap U_{\{1,i\}}) = 0$ guarantees that $\dim(\widetilde{U}_{\{3,i\}}) = \dim(U_{\{1,i\}}) = k$.

To prove that the assignment $(\widetilde{U}_B)_{B \in V'}$ forms an orthogonal $k$-subspace representation of the graph, let $B_1$ and $B_2$ be disjoint sets in $V'$.
If $3 \notin B_1 \cup B_2$ then we have $\widetilde{U}_{B_1} = U_{B_1}$ and $\widetilde{U}_{B_2} = U_{B_2}$, so it is clear that $\widetilde{U}_{B_1}$ and $\widetilde{U}_{B_2}$ are orthogonal.
Otherwise, assume without loss of generality that $B_1 = \{3,i\}$ for some $i \geq 4$ and that $3 \notin B_2$.
In this case we have $\widetilde{U}_{B_2} = U_{B_2}$, and since $B_2$ is disjoint from $B_1$ it is also disjoint from $\{1,i\}$ and from $\{1,3\}$, hence $\widetilde{U}_{B_2}$ is orthogonal to both $U_{\{1,i\}}$ and $U_{\{1,3\}}$ as well as to the projection $\widetilde{U}_{B_1}$ of $U_{\{1,i\}}$ to the subspace orthogonal to $U \subseteq U_{\{1,3\}}$.
We get that $\widetilde{U}_{B_1}$ and $\widetilde{U}_{B_2}$ are orthogonal, as required.

Finally, observe that all the subspaces $\widetilde{U}_{B}$ lie in the subspace of $\R^t$ that is orthogonal to $U$.
Indeed, for sets $B$ with $3 \in B$ this follows from the definition of $\widetilde{U}_{B}$, and for the other sets this holds because they are disjoint from $\{1,3\}$.
By applying an orthogonal linear transformation from this subspace to $\R^{t-\dim(U)}$, we obtain that
\[ \od_k(K^{-}(d-1,2),\R) \leq t- \dim(U) \leq t- \Big\lceil \frac{k}{2} \Big\rceil,\]
and as in the previous case, by the induction hypothesis it follows that
$t \geq \lceil \frac{k}{2} \rceil \cdot (d-4)+2k$, completing the proof.
\end{proof}

\subsection{General $s$}

We now prove Theorem~\ref{thm:Intro_general_s} which provides a lower bound on the generalized orthogonality dimension parameters of Kneser graphs $K(d,s)$ for $s \geq 3$.


\begin{proof}[ of Theorem~\ref{thm:Intro_general_s}]
Fix integers $k \geq s \geq 3$ and denote $m = \lceil \frac{k+1}{s} \rceil$.
Let $d_0=d_0(s,k)$ be a sufficiently large integer to be determined later.
We apply an induction on $d$.
To do so, we define $c=c(s,k)$ to be sufficiently large, say, $c=\frac{k-m+1}{s-1} \cdot (d_0+s-2)$, so that the statement of the theorem trivially holds for all integers $d \leq d_0+s-2$, and turn to prove the statement for $d \geq d_0$ assuming that it holds for $d-(s-1)$.

Let $(U_A)_{A \in V}$ be a $t$-dimensional orthogonal $k$-subspace representation of $K(d,s)$.
We start with some notation.
For an $s$-subset $A$ of $[d]$, an element $i \in A$, and an $s$-subset $B$ of $[d]$ satisfying $A \cap B = \{i\}$, we let $\calG_{A,i}(B)$ denote the collection that consists of the set $B$ and all the sets obtained from $B$ by replacing $i$ with some element from $A \setminus \{i\}$. Note that $|\calG_{A,i}(B)|=s$.
We say that a vertex $A$ of $K(d,s)$ is {\em good} (with respect to the given orthogonal subspace representation) if there exists an $i \in A$ such that for every vertex $B$ satisfying $A \cap B = \{i\}$ it holds that $\dim(U_A \cap U_C) \leq m-1$ for some $C \in \calG_{A,i}(B)$.

Assume first that there exists a good vertex $A$ in $K(d,s)$ associated with an element $i \in A$.
Applying Lemma~\ref{lem_cor:good_subspace}, we get that there exists a $(k-m+1)$-subspace $U$ of $U_A$ such that for every vertex $B$ satisfying $A \cap B = \{i\}$ it holds that $\dim(U \cap U_C) =0$ for some $C \in \calG_{A,i}(B)$.
We define an orthogonal $k$-subspace representation of the graph $K(d-(s-1),s)$ on the ground set $[d] \setminus (A \setminus \{i\})$ as follows. Let $B$ be an $s$-subset of $[d] \setminus (A \setminus \{i\})$.
If $i \notin B$ we define $\widetilde{U}_B = U_B$. Otherwise, we have $A \cap B = \{i\}$, and we let $\widetilde{U}_B$ be the projection of $U_C$ to the subspace of $\R^t$ orthogonal to $U$, where $C \in \calG_{A,i}(B)$ is a set satisfying $\dim(U \cap U_C) =0$. Note that this condition guarantees that $\dim(\widetilde{U}_B) = \dim(U_C) = k$.

We claim that the subspaces $\widetilde{U}_B$ form an orthogonal $k$-subspace representation of the graph $K(d-(s-1),s)$.
To see this, let $B_1$ and $B_2$ be disjoint $s$-subsets of $[d] \setminus (A \setminus \{i\})$.
If $i \notin B_1 \cup B_2$ then we have $\widetilde{U}_{B_1} = U_{B_1}$ and $\widetilde{U}_{B_2} = U_{B_2}$, so it is clear that $\widetilde{U}_{B_1}$ and $\widetilde{U}_{B_2}$ are orthogonal. Otherwise, assume without loss of generality that $i \in B_1$ and $i \notin B_2$. In this case, $\widetilde{U}_{B_2} = U_{B_2}$, and $\widetilde{U}_{B_1}$ is the projection of $U_C$ to the subspace of $\R^t$ orthogonal to $U$ for some $C \in \calG_{A,i}(B_1)$.
Since $B_2$ is disjoint from $A$, it follows that the subspace $\widetilde{U}_{B_2}$ is orthogonal to $U_A$ as well as to its subspace $U$.
It also follows that $B_2$ is disjoint from every set in $\calG_{A,i}(B_1)$, hence the subspace $\widetilde{U}_{B_2}$ is orthogonal to $U_C$. We get that $\widetilde{U}_{B_2}$ is orthogonal to $\widetilde{U}_{B_1}$, as required.

Now, observe that the above orthogonal $k$-subspace representation of $K(d-(s-1),s)$ lies in the subspace of $\R^t$ that is orthogonal to the $(k-m+1)$-subspace $U$.
Indeed, for sets $B$ with $i \in B$ this follows from the definition of $\widetilde{U}_{B}$, and for the other sets this holds because they are disjoint from $A$.
By applying an orthogonal linear transformation from this subspace to $\R^{t-\dim(U)}$, it follows that \[\od_{k}(K(d-(s-1),s),\R) \leq t-\dim(U) = t-(k-m+1).\]
Using the induction hypothesis, this implies that
\[t \geq \frac{k-m+1}{s-1} \cdot (d-(s-1)) -c + (k-m+1)= \frac{k-m+1}{s-1} \cdot d-c,\]
and we are done.

We are left with the case where no vertex of $K(d,s)$ is good, for which we need the following lemma.

\begin{lemma}\label{lemma:bad_vertex}
If a vertex $A$ of $K(d,s)$ is not good then there exists a nonzero vector $u_A \in U_A$ such that the number of vertices $D$ of $K(d,s)$ for which $U_D$ is not orthogonal to $u_A$ is at most
\[\binom{2s-1}{2} \cdot \binom{d-2}{s-2}.\]
\end{lemma}

We first show how Lemma~\ref{lemma:bad_vertex} completes the proof of the theorem.
Assume that no vertex of $K(d,s)$ is good, and consider the  following process: We start with the entire vertex set of $K(d,s)$, and in every iteration we choose an arbitrary vertex $A$ associated with its nonzero vector $u_A \in U_A$ from Lemma~\ref{lemma:bad_vertex} and eliminate all vertices whose subspaces are not orthogonal to $u_A$. The nonzero vectors associated with the chosen vertices are clearly pairwise orthogonal, and their number, just like the number of iterations in the process, is at least
\[\frac{\binom{d}{s}}{\binom{2s-1}{2} \cdot \binom{d-2}{s-2}} \geq \frac{k-m+1}{s-1} \cdot d-c,\]
where the inequality holds for every $d \geq d_0$ assuming that $d_0 = d_0(s,k)$ is sufficiently large (because the left-hand side of the inequality is quadratic in $d$ whereas the right-hand side is linear in $d$).
However, the size of the obtained orthogonal set cannot exceed the dimension $t$, hence
\[t \geq \frac{k-m+1}{s-1} \cdot d-c,\]
and we are done. It remains to prove Lemma~\ref{lemma:bad_vertex}.

\begin{proof}[ of Lemma~\ref{lemma:bad_vertex}]
Assume that $A$ is not a good vertex of $K(d,s)$ and fix an arbitrary $i \in A$. Then there exists a vertex $B$ satisfying $A \cap B = \{i\}$ such that $\dim(U_A \cap U_C) \geq m$ for every vertex $C \in \calG_{A,i}(B)$.
Denote $\calG_{A,i}(B) = \{C_1,\ldots,C_s\}$, and recall that every set $C_j$ intersects $A$ at one distinct element. For every $j \in [s]$ define $V_j = U_A \cap U_{C_j}$ and $W_j = V_1 + \cdots + V_j$. Note that
\[W_1 \subseteq W_2 \subseteq \cdots \subseteq W_s \subseteq U_A.\]
Since $\dim(W_1) = \dim(V_1) \geq m$ and $\dim(U_A)= k < m \cdot s$, there must exist some $j \in [s-1]$ for which
\begin{eqnarray}\label{eq:<ell}
\dim(W_{j+1}) - \dim(W_j) < m.
\end{eqnarray}
For this $j$, we have
\begin{eqnarray*}
\dim(V_{j+1} \cap W_j) &=& \dim(V_{j+1})+ \dim(W_j)- \dim(V_{j+1}+W_j)
\\ &=& \dim(V_{j+1})+ \dim(W_j) - \dim(W_{j+1}) > m-m=0,
\end{eqnarray*}
where the inequality follows by combining~\eqref{eq:<ell} with the fact that $\dim(V_{j+1}) \geq m$.
This implies that there exists a nonzero vector $u_A$ in $V_{j+1} \cap W_j$.
Observe that
\[u_A \in U_A \cap U_{C_{j+1}} \cap (U_{C_1} +\cdots+U_{C_j}).\]
Now, consider a vertex $D$ of $K(d,s)$ whose subspace $U_D$ is not orthogonal to $u_A$.
It follows that $U_D$ is not orthogonal to $U_A$, to $U_{C_{j+1}}$, and to $U_{C_1} +\cdots+U_{C_j}$, hence $D$ intersects the sets $A$ and $C_{j+1}$ as well as at least one of the sets $C_1,\ldots,C_{j}$.
We claim that $D$ must include at least two elements from $A \cup B$.
Indeed, $D$ intersects $A$ but if $D$ includes from $A \cup B$ only one element and this element belongs to $A$ then $D$ either does not intersect $C_{j+1}$ or does not intersect any of $C_1,\ldots,C_{j}$.
It follows that the number of vertices $D$ for which $U_D$ is not orthogonal to $u_A$ is bounded from above by the number of $s$-subsets of $[d]$ that include at least two elements from the $2s-1$ elements of $A \cup B$. The latter is at most $\binom{2s-1}{2} \cdot \binom{d-2}{s-2}$, as required.
\end{proof}
The proof is completed.
\end{proof}

As immediate corollaries of Theorem~\ref{thm:Intro_general_s}, we obtain the following.
\begin{corollary}\label{cor:ls-1}
For every integers $s \geq 3$ and $\ell \geq 2$ there exists $c=c(s,\ell)$ such that for all integers $d \geq 2s$,
\[\od_{\ell \cdot s-1}(K(d,s),\R) \geq \ell \cdot d-c.\]
\end{corollary}
\noindent
As mentioned before, the bound given in Corollary~\ref{cor:ls-1} is tight up to the additive constant $c$.

\begin{corollary}\label{cor:3/2}
There exists an absolute constant $c$ such that for all integers $d \geq 6$,
\[\od_{4}(K(d,3),\R) \geq 3d/2-c.\]
\end{corollary}

Equipped with Corollary~\ref{cor:3/2}, we are ready to deduce Theorem~\ref{thm:IntroHardness}.

\begin{proof}[ of Theorem~\ref{thm:IntroHardness}]
Let $t$ be a sufficiently large integer. Recall that a result of~\cite{BukhC18} implies that $\od_3(K(t,3),\R) = t$, whereas Corollary~\ref{cor:3/2} implies that $\od_{4}(K(t,3),\R) \geq 3t/2-c$ for an absolute constant $c$.
Applying Proposition~\ref{prop:xi3vs4} with $F = K(t,3)$, it follows that it is $\NP$-hard to decide whether an input graph $G$ satisfies $\od(G,\R) \leq t$ or $\od(G,\R) \geq 3t/2-c$, as desired.
\end{proof}

\section{The Minrank of Generalized Kneser Graphs}\label{sec:ODofGenKneser}

In this section we consider a generalization of the family of Kneser graphs, defined as follows.


\begin{definition}[Generalized Kneser Graphs]\label{def:GenKneser}
For integers $m \leq s \leq d$, the {\em generalized Kneser graph} $\Kneser{d}{s}{m}$ is the graph whose vertices are all the $s$-subsets of $[d]$, where two sets $A,B$ are adjacent if $|A \cap B| < m$.
\end{definition}

For this family of graphs, we prove the following upper bound on the minrank parameter over finite fields (recall Definition~\ref{def:minrank}).

\begin{theorem}\label{thm:minrk_Kneser}
For all integers $m \leq s \leq d$ and for every finite field $\Fset$,
\[{\minrank}_{\Fset}(\Kneser{d}{s}{m}) \leq \sum_{i=0}^{s-m}{d \choose i}.\]
Moreover, the bound on the minrank can be achieved by a symmetric matrix.
\end{theorem}

\begin{remark}\label{remark:symmetric}
Theorem~\ref{thm:minrk_Kneser} guarantees that the bound on the minrank can be achieved by a symmetric matrix.
This will be crucial for one of our applications, namely, for a construction of triangle-free graphs whose complement has low orthogonality dimension over the binary field $\Fset_2$ (see Section~\ref{sec:OD_F_2_triangle}).
We remark, however, that for undirected graphs and for fields of characteristic different from $2$, attaining the bound on the minrank by a symmetric matrix can be achieved easily with a factor of $2$ worse bound on the minrank. Indeed, if a matrix $M$ represents a graph $G$ over a field $\Fset$ of characteristic different from $2$ and satisfies ${\rank}_\Fset(M) = r$ then the matrix $M+M^T$ also represents $G$ and has rank at most $2r$ over $\Fset$.
This argument does not hold over fields of characteristic $2$, since in this case the diagonal entries of $M+M^T$ are all zeros.
\end{remark}

As in the previous section, we start with a simple linear algebra lemma.

\subsection{Linear Algebra Lemma}

\begin{lemma}\label{lemma:rankp_R}
For a graph $G$ on the vertex set $[n]$, let $M \in \Z^{n \times n}$ be an integer matrix such that $M_{i,i}=1$ for every $i \in [n]$, and $M_{i,j} = 0$ for every distinct non-adjacent vertices $i$ and $j$ in $G$. Then, for every finite field $\Fset$, ${\minrank}_{\Fset}(G) \leq \rank_\R(M)$.
\end{lemma}

We need the following fact.

\begin{fact}\label{fact:rankp_R}
Let $p$ be a prime and let $M$ be an integer matrix. Then, the matrix $M' = M~(\mod~p)$ satisfies ${\rank}_{\Fset_p}(M') \leq {\rank}_\R(M)$.
\end{fact}

\begin{proof}
It suffices to show that if some rows $v_1,\ldots, v_k$ of $M$ are linearly dependent over $\R$ then, considered modulo $p$, they are also linearly dependent over $\Fset_p$.
To see this, assume that there exist $a_1,\ldots,a_k \in \R$, at least one of which is nonzero, for which $\sum_{i=1}^{k}{a_i v_i}=0$. Since the $v_i$'s are integer vectors it can be assumed that $a_1,\ldots,a_k \in \Z$ and that $\gcd(a_1,\ldots,a_k)=1$. This implies that they are not all zeros modulo $p$. Therefore, the same coefficients, considered modulo $p$, provide a non-trivial combination of the corresponding rows of $M'$ with sum zero, and we are done.
\end{proof}

\begin{proof}[ of Lemma~\ref{lemma:rankp_R}]
Let $\Fset$ be a finite field and denote its characteristic by $p$.
For a graph $G$ and an integer matrix $M$ as in the statement of the lemma, consider the matrix $M' = M~(\mod~p)$.
Observe that $M'$ represents $G$ over $\Fset_p$, hence by Fact~\ref{fact:rankp_R},
${\minrank}_{\Fset_p}(G) \leq \rank_{\Fset_p}(M') \leq \rank_\R(M)$.
Since $\Fset_p$ is a subfield of $\Fset$, it holds that ${\minrank}_{\Fset}(G) \leq {\minrank}_{\Fset_p}(G)$, and we are done.
\end{proof}

\subsection{Proof of Theorem~\ref{thm:minrk_Kneser}}

We are ready to prove Theorem~\ref{thm:minrk_Kneser}.

\begin{proof}[ of Theorem~\ref{thm:minrk_Kneser}]
Consider the polynomial $q \in \R[x]$ defined by
\[q(x) = {\binom {x-m} {s-m}} = \frac{1}{(s-m)!} \cdot (x-m)(x-(m+1)) \cdots (x-(s-1)).\]
Notice that $q$ is an integer-valued polynomial of degree $s-m$.
Let $f: \{0,1\}^d \times \{0,1\}^d \rightarrow \R$ be the function defined by
\[ f(x,y) = q  \Big ( \sum_{i=1}^{d}{x_i y_i}\Big )\]
for every $x,y \in \{0,1\}^d$.
Expanding $f$ as a linear combination of monomials, the relation $z^2 = z$ for $z \in \{0,1\}$ implies that one can reduce to $1$ the exponent of each variable occuring in a monomial. It follows that $f$ can be represented as a multilinear polynomial in the $2d$ variables of $x$ and $y$. By combining terms involving the same monomial in the variables of $x$, one can write $f$ as
\[ f(x,y) = \sum_{i=1}^{R}{g_i(x) h_i(y)} \]
for an integer $R$ and functions $g_i, h_i : \{0,1\}^d \rightarrow \R$, $i \in [R]$, such that the $g_i$'s are distinct multilinear monomials of total degree at most $s-m$ in $d$ variables. It follows that $R \leq \sum_{i=0}^{s-m}{d \choose i}$.

Now, let $M_1$ and $M_2$ be the $2^d \times R$ matrices whose rows are indexed by $\{0,1\}^d$ and whose columns are indexed by $[R]$, defined by $(M_1)_{x,i} = g_i(x)$ and $(M_2)_{x,i} = h_i(x)$. Then, the rank over $\R$ of the matrix $M = M_1 \cdot M_2^T$ is at most $R$ and for every $x,y \in\{0,1\}^d$ it holds that $M_{x,y} = f(x,y)$. By the definition of $f$ the matrix $M$ is symmetric, and since $q$ is an integer-valued polynomial, all of its entries are integer.

Finally, let $V$ be the vertex set of $\Kneser{d}{s}{m}$, that is, the collection of all $s$-subsets of $[d]$, and identify every vertex $A \in V$ with an indicator vector $c_A \in \{0,1\}^d$ in the natural way.
Observe that for every $A,B \in V$ we have
\[M_{c_A, c_B} = f(c_A,c_B) = q(|A \cap B|).\]
Hence, for every $A \in V$ we have $|A|=s$ and thus $M_{c_A,c_A} =q(s) = 1$, whereas for every distinct non-adjacent $A,B \in V$ we have $m \leq |A \cap B|\leq s-1$ and thus $M_{c_A,c_B} = q(|A \cap B|) =0$.
Since the restriction of $M$ to $V \times V$ is symmetric and has rank at most $R$ over the reals, Lemma~\ref{lemma:rankp_R} implies that ${\minrank}_{\Fset}(\Kneser{d}{s}{m}) \leq R$ for every finite field $\Fset$ and that the bound can be achieved by a symmetric matrix, as desired.
\end{proof}

\subsection{Applications}

We gather below several applications of Theorem~\ref{thm:minrk_Kneser}.

\subsubsection{The Odd Alternating Cycle Conjecture over Finite Fields}

We turn to disprove Conjecture~\ref{conj:alternating} over every finite field.
We will use the simple fact that generalized Kneser graphs do not contain short odd cycles, as stated below (see, e.g.,~\cite{Denley97,Haviv18free}).

\begin{lemma}\label{lemma:cycle_K}
Let $\ell \geq 3$ be an odd integer.
For every even integer $d$ and an integer $m \leq \frac{d}{2\ell}$, the graph $\Kneser{d}{\frac{d}{2}}{m}$ contains no odd cycle of length at most $\ell$.
\end{lemma}

We prove the following theorem, confirming Theorem~\ref{thm:Intro_Cycles}.

\begin{theorem}\label{thm:Cycles}
For every odd integer $\ell \geq 3$ there exists $\delta = \delta(\ell) >0$ such that for every sufficiently large integer $n$, there exists an $n$-vertex graph $G$ with no odd cycle of length at most $\ell$ such that for every finite field $\Fset$,
\[{\minrank}_{\Fset}(G) \leq n^{1-\delta}.\]
Moreover, the bound on the minrank can be achieved by a symmetric matrix.
\end{theorem}

\begin{proof}
Fix an odd integer $\ell \geq 3$.
For an integer $d$ divisible by $2 \ell$, consider the graph $G = \Kneser{d}{\frac{d}{2}}{m}$ where $m = \frac{d}{2 \ell}$.
By Lemma~\ref{lemma:cycle_K}, $G$ contains no odd cycle of length at most $\ell$.
As for the minrank parameter, Theorem~\ref{thm:minrk_Kneser} implies that for every finite field $\Fset$,
\[{\minrank}_{\Fset}(G) \leq \sum_{i=0}^{d/2-m}{d \choose i} \leq 2^{H(\frac{1}{2}-\frac{m}{d}) \cdot d} = 2^{H(\frac{1}{2}-\frac{1}{2\ell}) \cdot d},\]
where $H$ stands for the binary entropy function.
Since $G$ has $|V| = {d \choose {d/2}} = 2^{(1-o(1)) \cdot d}$ vertices, for any $\delta>0$ such that $H(\frac{1}{2}-\frac{1}{2\ell}) < 1-\delta$ we have ${\minrank}_{\Fset}(G) \leq |V|^{1-\delta}$ for every sufficiently large integer $d$.
The proof is completed by considering, for every sufficiently large integer $n$, some $n$-vertex subgraph of the graph defined above, where $d$ is the smallest integer divisible by $2\ell$ such that $n \leq {d \choose {d/2}}$.
\end{proof}

\subsubsection{Triangle-free Graphs and the Orthogonality Dimension over the Binary Field}\label{sec:OD_F_2_triangle}

We turn to prove Theorem~\ref{thm:IntroTriangleFree}.
Its proof adopts the following special case of a result due to Lempel~\cite{Lempel75}.
\begin{lemma}[\cite{Lempel75}]\label{lem:lempel}
Let $M$ by an $n$ by $n$ symmetric matrix over the binary field $\Fset_2$ with at least one nonzero diagonal entry and rank $r$. Then, there exists an $n$ by $r$ matrix $B$ over $\Fset_2$ satisfying $M = B \cdot B^T$.
\end{lemma}

\begin{proof}[ of Theorem~\ref{thm:IntroTriangleFree}]
Apply Theorem~\ref{thm:Cycles} with $\ell = 3$ to obtain some $\delta >0$ such that for every sufficiently large integer $n$, there exist a triangle-free $n$-vertex graph $G$ and an $n$ by $n$ symmetric matrix $M$ over $\Fset_2$ of rank $r = {\rank}_{\Fset_2}(M) \leq n^{1-\delta}$ that represents $G$.
By Lemma~\ref{lem:lempel}, there exists an $n$ by $r$ matrix $B$ over $\Fset_2$ satisfying $M = B \cdot B^T$.
By assigning the $i$th row of $B$ to the $i$th vertex of $G$ we get an $r$-dimensional orthogonal representation of $\overline{G}$ over $\Fset_2$, hence $\od(\overline{G},\Fset_2) \leq r \leq n^{1-\delta}$.
\end{proof}

\subsubsection{The Vector Chromatic Number vs. Minrank}

The vector chromatic number of graphs, introduced by Karger, Motwani, and Sudan in~\cite{KargerMS98}, is defined as follows.

\begin{definition}[Vector Chromatic Number]\label{def:chi_v}
For a graph $G=(V,E)$ the {\em vector chromatic number} of $G$,
denoted by $\vchrom(G)$, is the minimal real value of $\kappa > 1$ such
that there exists an assignment of a unit vector $w_v$ to every
vertex $v \in V$ satisfying the inequality $\langle w_{v}, w_{v'} \rangle \leq
-\frac{1}{\kappa -1}$ whenever $v$ and $v'$ are adjacent in $G$.
\end{definition}

To prove Theorem~\ref{thm:Intro_chi_v}, we need the following simple fact that relates the minrank of a graph to the minrank of its complement (see, e.g.,~\cite[Remark~2.2]{Peeters96}).

\begin{fact}\label{fact:minrank_comp}
For every field $\Fset$ and an $n$-vertex graph $G$, $\minrank_\Fset(G) \cdot \minrank_\Fset(\overline{G}) \geq n$.
\end{fact}

\begin{proof}[ of Theorem~\ref{thm:Intro_chi_v}]
For an integer $d$ divisible by $8$, consider the graph $G = \Kneser{d}{\frac{d}{2}}{m}$ where $m = \frac{d}{8}$.
We first claim that $\vchrom(G) \leq 3$.
To see this, assign to every vertex $A$ of $G$, representing a $\frac{d}{2}$-subset of $[d]$, the unit vector $w_A \in \R^d$ defined by $(w_A)_i = \frac{1}{\sqrt{d}}$ if $i \in A$ and $(w_A)_i = - \frac{1}{\sqrt{d}}$ otherwise.
Observe that every two distinct vertices $A$ and $B$ that are adjacent in $G$ satisfy $|A \cap B| < \frac{d}{8}$ and thus $|A \bigtriangleup B| > \frac{3d}{4}$, implying that $\langle w_A, w_B \rangle = \frac{d-2 \cdot |A \bigtriangleup B|}{d} < -\frac{1}{2}$. This implies that $\vchrom(G) \leq 3$, as claimed.
As for the minrank parameter, Theorem~\ref{thm:minrk_Kneser} implies that for every finite field $\Fset$,
\[{\minrank}_{\Fset}(G) \leq \sum_{i=0}^{d/2-m}{d \choose i} \leq 2^{H(\frac{1}{2}-\frac{m}{d}) \cdot d} = 2^{H(3/8) \cdot d},\]
where $H$ stands for the binary entropy function.
Since $G$ has $n = {d \choose {d/2}} = 2^{(1-o(1)) \cdot d}$ vertices, for any $\delta < 1- H(3/8)$ we have ${\minrank}_{\Fset}(G) \leq n^{1-\delta}$ assuming that $d$ is sufficiently large. By Fact~\ref{fact:minrank_comp}, this implies that
${\minrank}_{\Fset}(\overline{G}) \geq n^{\delta}$, and we are done.
\end{proof}


\bibliographystyle{abbrv}
\bibliography{xi_kneser}

\end{document}